\documentclass[11pt]{article}
\usepackage[utf8]{inputenc}
\usepackage{amsmath, amsthm, amssymb, amsfonts, braket, complexity, enumerate, mathrsfs, tikz, caption, subcaption, xcolor, mathtools, braket, dirtytalk}

\usepackage{tabularx}

\usepackage[linesnumbered,ruled,vlined]{algorithm2e}
\SetKwInput{KwInput}{Input}     
\SetKwInput{KwOutput}{Output}
\SetKwFor{RepTimes}{repeat}{times}{end}

\usepackage[pagebackref,colorlinks,citecolor=blue,linkcolor=magenta,bookmarks=true]{hyperref}
\usepackage[nameinlink,capitalize]{cleveref}

\newcommand{\nn}{\nonumber\\}

\newcommand{\wh}[1]{\widehat{#1}}

\newcommand{\F}{\mathbb{F}}
\renewcommand{\hat}{\widehat}

\theoremstyle{plain}
\newtheorem{theorem}{Theorem}[section]
\newtheorem{corollary}[theorem]{Corollary}
\newtheorem{definition}[theorem]{Definition}
\newtheorem{proposition}[theorem]{Proposition}
\newtheorem{lemma}[theorem]{Lemma}
\newtheorem{claim}[theorem]{Claim}
\newtheorem{fact}[theorem]{Fact}

\newtheorem{question}[theorem]{Question}
\newcommand{\stab}{\textrm{Stab}}
\newcommand{\stabset}{\mathcal{S}}

\newcommand{\eps}{\varepsilon}
\newcommand{\labs}{\left|}
\newcommand{\rabs}{\right|}
\newcommand{\lpar}{\left(}
\newcommand{\rpar}{\right)}
\newcommand{\lbrak}{\left[}
\newcommand{\rbrak}{\right]}

\newcommand{\extent}{\xi}
\newcommand{\fidelity}{F}

\renewcommand{\Pr}{\mathop{\bf Pr\/}}
\newcommand{\Ex}{\mathop{\bf E\/}}
\newcommand{\Es}[1]{\mathop{\bf E\/}_{{\substack{#1}}}}

\newcommand{\calC}{\mathcal{C}}

\newcommand{\calS}{\mathcal{S}}

\newcommand{\abs}[1]{\lvert #1 \rvert}

\newcommand{\norm}[1]{\lVert #1 \rVert}

\newcommand{\ketbra}[2]{\ket{#1}\!\!\bra{#2}}

\usepackage{amsmath,amssymb,amsthm,amsfonts,latexsym,bbm,xspace,graphicx,float,mathtools,braket,}
\usepackage{comment} 
\usepackage[letterpaper,margin=1in]{geometry}
\usepackage{enumitem}

\usepackage[colorinlistoftodos]{todonotes}

\renewcommand{\backref}[1]{}

\renewcommand{\backrefalt}[4]{%
\ifcase #1 %
\or
[p.\ #2]%
\else
[pp.\ #2]%
\fi}

\usepackage{apxproof}

\usepackage{amsmath}

\newcommand{\pauli}{\mathcal{P}_n}
\newcommand{\etamin}{\eta_{min}}

\title{Low-Stabilizer-Complexity Quantum States Are Not Pseudorandom}

\author{Sabee Grewal\thanks{The University of Texas at Austin. \texttt{\{sabee, kretsch\}@cs.utexas.edu}, \texttt{\{vishnu.iyer, dliang\}@utexas.edu}.}\and Vishnu Iyer\footnotemark[1]\and William Kretschmer\footnotemark[1]\and Daniel Liang\footnotemark[1]}

\date{}

\begin{document}

\maketitle
\begin{abstract}

We show that quantum states with ``low stabilizer complexity'' can be efficiently distinguished from Haar-random. 
Specifically, given an $n$-qubit pure state $\ket{\psi}$, we give an efficient algorithm that distinguishes whether $\ket{\psi}$ is (i) Haar-random or (ii) a state with stabilizer fidelity at least $\frac{1}{k}$ (i.e., has fidelity at least $\frac{1}{k}$ with some stabilizer state), promised that one of these is the case.
With black-box access to $\ket{\psi}$, our algorithm uses $O\!\left( k^{12} \log(1/\delta)\right)$ copies of $\ket{\psi}$ and $O\!\left(n k^{12} \log(1/\delta)\right)$ time to succeed with probability at least $1-\delta$, and, with access to a state preparation unitary for $\ket{\psi}$ (and its inverse), $O\!\left( k^{3} \log(1/\delta)\right)$ queries and $O\!\left(n k^{3} \log(1/\delta)\right)$ time suffice.

As a corollary, we prove that $\omega(\log(n))$ $T$-gates are necessary for any Clifford+$T$ circuit to prepare computationally pseudorandom quantum states, a first-of-its-kind lower bound.
\end{abstract}

\section{Introduction}
\label{sec:intro}
The stabilizer formalism \cite{gottesman1997stabilizer} plays a central role in quantum information. Stabilizer states are states that lie in the intersection of the positive eigenspaces of $2^n$ commuting Pauli operators. Stabilizer states can be generated by Clifford circuits, which are the group of unitary transformations that normalize the Pauli group. Stabilizer states and the Clifford group have widespread applications in quantum error correction \cite{shor1995codes, calderbank1996codes}, measurement-based quantum computation \cite{raussendorf2000mbqc}, randomized benchmarking \cite{knill2008benchmarking}, and quantum learning algorithms \cite{huangkuengpresskill}. 
These applications are largely thanks to the rich algebraic structure afforded by the stabilizer formalism. 

Stabilizer states are also one of the few classes of states that admit efficient learning algorithms.
Montanaro \cite{montanaro2017learning} gave an algorithm that takes $O(n)$ copies of an $n$-qubit stabilizer state and correctly identifies the state with high probability in time $O(n^3)$.
Gross, Nezami, and Walter \cite{gross2021schur} gave an algorithm for \textit{property testing} stabilizer states, which is the task of distinguishing whether a state is a stabilizer state or is far from any stabilizer state. Remarkably, this algorithm requires only $6$ copes of the state.

Despite finding numerous applications, Clifford circuits are not universal for quantum computation. Furthermore, in 1998, Gottesman and Knill showed that Clifford circuits acting on stabilizer states can be efficiently classically simulated \cite{gottesman1998heisenberg, aaronson2004stabilizer}. 
However, with the additional ability to apply a $T$-gate (the gate $\ket{0}\!\!\bra{0} + e^{i \pi / 4}\ket{1}\!\!\bra{1}$), the resulting gate set becomes universal.
Therefore, efficient simulation of so-called Clifford+$T$ circuits would imply $\BPP = \BQP$, and a large line of work has been devoted to developing better simulation algorithms \cite{PashayanPhysRevLett.115.070501, BrayviPhysRevLett.116.250501, RallPhysRevA.99.062337, Bravyi2019simulationofquantum}.

Currently, the best-performing simulation algorithms are based on modeling the output state of a quantum circuit as a decomposition of stabilizer states \cite{Bravyi2019simulationofquantum}. 
These decompositions give rise to simulation algorithms whose runtimes scale polynomially in the complexity of the decomposition.
One such complexity measure is the \textit{stabilizer extent}.
Consider the state $\ket{\psi} = \sum_i c_i \ket{\phi_i}$ for $c_i \in \mathbb{C}$ and stabilizer states $\ket{\phi_i}$. 
The stabilizer extent is the minimum $\lpar\sum_i{\abs{c_i}}\rpar^2$ over all such decompositions of $\ket{\psi}$, and scales exponentially in the number of $T$-gates in the circuit producing the state.
A closely-related complexity measure is the \textit{stabilizer fidelity}, which is the maximum overlap between $\ket{\psi}$ and any stabilizer state. 
Indeed, the inverse of stabilizer fidelity lower bounds stabilizer extent \cite{Bravyi2019simulationofquantum}. Collectively, we informally refer to states with either low stabilizer extent or non-negligible stabilizer fidelity as states of low ``stabilizer complexity''.

As a generalization of stabilizer states, it is natural to ask whether states of low stabilizer complexity are also efficiently learnable, and indeed a similar question has been raised before \cite{arunachalam2022phase}. Nevertheless, this problem remains largely open except in some highly restricted settings \cite{lai2022learning}. This could be in part because many of the useful properties of stabilizer states provably fail to generalize to states with low stabilizer complexity. 
For example, \cite{hinsche2022learning} observed that one can efficiently learn the output distribution of any Clifford circuit, given samples from this distribution.\footnote{Indeed, every such distribution is simply an affine subspace of $\mathbb{F}_2^n$.}
However, this task already becomes intractable for circuits with a \textit{single} $T$-gate (producing a state of constant stabilizer extent), where \cite{hinsche2022learning} proved that learning the output distribution is as hard as the learning parities with noise problem. 

Furthermore, it is known that stabilizer states form a \textit{$t$-design} for $t = 3$, meaning that random stabilizer states duplicate the first 3 moments of the Haar measure \cite{ kuenghttps://doi.org/10.48550/arxiv.1510.02767, webb2016clifford}. 
By contrast, \cite{haferkamp2020homeopathy} showed that circuits with $\poly(t)$ non-Clifford gates are sufficient to generate approximate $t$-designs. Thus, for any constant $t$, states of constant stabilizer extent can form approximate $t$-designs. This suggests that states of low stabilizer complexity can give much stronger information-theoretic approximations to the Haar measure than ordinary stabilizer states, because stabilizer states fail to form a $t$-design for any $t > 3$ \cite{ZKGG16}.

In this work, we investigate whether these properties that differentiate stabilizer states from low-stabilizer-complexity states can be pushed further, to prove hardness of learning low-stabilizer-complexity states. 
One natural approach towards proving that low-stabilizer-complexity states are hard to learn would be to show that they are \textit{pseudorandom}. Ji, Liu, and Song \cite{Ji10.1007/978-3-319-96878-0_5} define an ensemble of $n$-qubit states to be (computationally) pseudorandom if every $\poly(n)$-time quantum adversary has at most a negligible advantage in distinguishing copies of a state drawn randomly from the ensemble from copies of a Haar-random $n$-qubit state. Note that pseudorandom states are not efficiently learnable, as any algorithm for learning some set of quantum states gives an algorithm to distinguish those states from the Haar measure.

Our main result is an efficient algorithm for distinguishing states of non-negligible stabilizer fidelity from Haar-random states, showing that such states \textit{cannot} be pseudorandom. This type of distinguishing is sometimes known as \textit{weak learning} in learning theory. 

\begin{theorem}[Informal version of \cref{thm:main-thm-alg}]\label{thm:main-theorem-informal}
Let $\ket{\psi}$ be an unknown $n$-qubit pure state, and let $k \leq \frac{4}{5}2^{n/12}$.
There is an efficient algorithm that distinguishes whether $\ket{\psi}$ is Haar-random or a state with stabilizer fidelity at least $\frac{1}{k}$, promised that one of these is the case. In particular, the algorithm uses $O(k^{12} \log(1/\delta))$ copies of $\ket{\psi}$ and $O(n k^{12} \log(1/\delta))$ time to succeed with probability at least $1 - \delta$.
\end{theorem}

\cref{thm:main-theorem-informal} also generalizes to distinguishing states with low stabilizer extent from Haar-random. 
To the best of our knowledge, prior to our work, it was even unknown whether states of stabilizer extent at most a \textit{constant} could be efficiently distinguished from Haar-random. We also emphasize that the contrast between our positive learning result and the hardness result of \cite{hinsche2022learning} stems in part from the differing access models: we assume access to copies of the quantum state, whereas \cite{hinsche2022learning} considers algorithms that only have outcomes of standard basis measurements of the state.

As a simple corollary, we prove a first-of-its-kind lower bound on the number of $T$-gates required to prepare computationally pseudorandom quantum states.

\begin{corollary}[\cref{cor:main-thm-pseudo}]\label{cor:main-thm-pseudo-informal}
Any family of Clifford+$T$ circuits that produces an ensemble of $n$-qubit computationally pseudorandom quantum states must use at least $\omega(\log n)$ $T$-gates.
\end{corollary}

In some sense, \cref{cor:main-thm-pseudo-informal} contrasts sharply with the result of \cite{haferkamp2020homeopathy}, where circuits containing just a few non-Clifford gates are sufficient to produce strong information-theoretic approximations to the Haar measure (i.e. $t$-designs). Nevertheless, we emphasize that our result and \cite{haferkamp2020homeopathy} are formally incomparable, because computationally pseudorandom states need not form approximate $t$-designs for constant $t$, nor vice-versa.

\subsection{Main Ideas}
Let $x = (p,q) \in \F_2^{2n}$, where $p$ and $q$ are the first and last $n$ bits of $x$, respectively. 
Define $W_x \coloneqq i^{p\cdot q}X^p Z^q$ (a Pauli operator without phase), and let $\ket{\Phi^+} \coloneqq 2^{-n/2} \sum_{x \in \F_2^n} \ket{x, x}$ be a maximimally entangled state. 
Then, the set $\{\ket{W_x} \coloneqq (W_x \otimes I)\ket{\Phi^+} \mid x \in \mathbb{F}_2^{2n}\}$ is the \textit{Bell basis}, an orthonormal basis of $\mathbb{C}^{2^n} \otimes \mathbb{C}^{2^n}$. 

Our algorithm uses \textit{Bell difference sampling} \cite{montanaro2017learning, gross2021schur}, which works as follows (see \cref{ssec:weyl_and_bell} for more detail): Given four copies of an $n$-qubit pure state $\ket{\psi}$, perform a Bell-basis measurement on $\ket{\psi}^{\otimes 2}$ to get measurement outcome $x \in \F_2^{2n}$, repeat this on the remaining two copies to get measurement outcome $y \in \F_2^{2n}$, and return $z = x + y$.

We refer to $p_\psi(x) \coloneqq 2^{-n} \abs{\braket{\psi|W_x|\psi}}^2$ as the \textit{characteristic distribution of} $\ket{\psi}$.
To see that $p_\psi$ is a distribution, recall that since the Pauli operators form an orthonormal basis over Hermitian matrices, we can always decompose $\ketbra{\psi}{\psi} = \frac{1}{2^{n}}\sum_{x \in \F_2^n} \braket{\psi | W_x | \psi} \cdot W_x$.
By assumption, $\abs{\braket{\psi|\psi}}^2 = 1$, so by Parseval's identity, \[\frac{1}{2^n}\sum_{x \in \F_2^n} \abs{\braket{\psi | W_x | \psi}}^2 = 1.\]

Gross, Nezami, and Walter\ \cite{gross2021schur} showed that Bell difference sampling an arbitrary pure state $\ket{\psi}$ corresponds to sampling a random operator $W_x$ according to the following distribution:
\[
q_\psi(x) =\!\!\! \sum_{y \in \F_2^{2n}} p_\psi (y) p_\psi (x + y).
\]
We call $q_\psi$ the \textit{Weyl distribution of} $\ket{\psi}$. Note that the Weyl distribution of $\ket{\psi}$ is the scaled convolution of the characteristic distribution with itself (i.e., $q_\psi = 4^n (p_\psi \ast p_\psi)$, where `$\ast$' is the convolution operator). 

Define the $\{\pm1\}$-outcome measurement $M_x \coloneqq \left\{\frac{I \pm W_x}{2}\right\}$ (projections onto the $\pm1$-eigenspaces of $W_x$).
Our algorithm begins by repeating the following process $m$ times: sample a random Weyl operator $W_{x}$ (via Bell difference sampling) and perform the measurement $M_x^{\otimes 2}$ on $\ket{\psi^{\otimes 2}}$. 
Then, average all of the measurement outcomes. If the average is at least $1/\poly(k)$, we decide that $\ket{\psi}$ has stabilizer fidelity at least $\frac{1}{k}$. Otherwise, we decide that $\ket{\psi}$ is Haar-random.

What statistic are we computing in our algorithm? 
Denote the measurement outcome on the $i$th iteration as $X_i \in \{\pm 1\}$. 
Observe that for all $X_i$,  
\[
\Es{}[X_i] 
=\!\!\! \sum_{x \in \F_2^{2n}} q_\psi(x) \abs{\braket{\psi|W_x|\psi}}^2  
= 2^n \!\!\!\sum_{x \in \F_2^{2n}} q_\psi(x) p_\psi(x)
= 2^n \Es{x \sim q_\psi}[p_\psi(x)],
\]
where the expectation $\Ex[X_i]$ is taken over sampling $x \sim q_\psi$ and the randomness from performing the measurement $M_x^{\otimes 2}$.
Hence, for our algorithm to work, $\Es{x \sim q_\psi}[p_\psi(x)]$ must be ``different enough'' when $\ket{\psi}$ either is Haar-random or has low stabilizer complexity. 
Proving that this is the case is the main technical ingredient of our work: 
\begin{lemma}[Informal version of \cref{lem:expected-measurement}]
Let $\ket{\psi}$ be an $n$-qubit pure state. Suppose the stabilizer fidelity of $\ket{\psi}$ is at least $\frac{1}{k}$. Then,
$$
2^n \Ex_{x \sim q_{\psi}}\lbrak p_\psi(x) \rbrak \geq \dfrac{1}{k^6}.
$$
In contrast, suppose $\ket{\psi}$ is a Haar-random quantum state. Then,
with overwhelming probability over the Haar measure,
$$
2^n \Ex_{x \sim q_{\psi}}\lbrak p_\psi(x) \rbrak \leq 2^{-n/2}.
$$
\end{lemma}
Our proof uses Fourier analysis of Boolean functions, and some parts of our proof are reminiscent of the celebrated Blum-Luby-Rubinfield linearity test \cite{BLRtest}.
Intuitively, $q_\psi$ is significantly closer to linear when $\ket{\psi}$ has non-negligible stabilizer fidelity, as opposed to when $\ket{\psi}$ is a Haar-random state. 

With the above lemma, all that remains is ``merely'' a sample complexity analysis, namely: what $m$ is sufficient to distinguish  whether the average is close to $0$ or $\Omega(1/k^6)$?
In the simplest case, we show that $O(k^{12}\log(1/\delta))$ samples are sufficient by Hoeffding's inequality. However, this complexity can be improved if given access to a unitary that prepares $\ket{\psi}$ (and its inverse). In this model, we are able to achieve a quartic speedup in both sample and time complexity, which we explain in \cref{sec:qae}.

\section{Preliminaries}
\label{sec:preliminaries}
First, we establish some notation used throughout this work. We denote $[n] \coloneqq \{1,\ldots,n\}$. For $v \in \mathbb{C}^n$, $\norm{v}_p \coloneqq (\sum_{i \in [n]}\abs{v_i}^p)^{1/p}$ is the $\ell_p$-norm. Logarithms are assumed to be in base $2$.
For a probability distribution $P$ on a set $S$, we denote drawing a sample $s \in S$ according to $P$ by $s \sim P$. 
We denote drawing a sample $s \in S$ uniformly at random by $s \sim S$.

\subsection{Stabilizer States and Stabilizer Complexity Measures}\label{subsec:stabilizer-prelim}
We define the $1$-qubit Pauli group to be the collection of matrices $\{I, X, Y, Z\}$, where
$$
I = 
\begin{pmatrix}
1 & 0 \\ 0 & 1
\end{pmatrix},\quad
X = 
\begin{pmatrix}
0 & 1 \\ 1 & 0
\end{pmatrix},\quad
Y = 
\begin{pmatrix}
0 & -i \\ i & 0
\end{pmatrix},\quad
Z = 
\begin{pmatrix}
1 & 0 \\ 0 & -1
\end{pmatrix}.
$$
The $n$-qubit Pauli group $\mathcal{P}_n$ is the set $\{\pm 1, \pm i\} \times \{I, X, Y, Z\}^{\otimes n}$. 
The Clifford group $\calC_n$ is the group of unitary transformations generated by $H$, $S$, and $\mathrm{CNOT}$ gates, where $H$ is the Hadamard gate, $S \coloneqq \ket{0}\!\!\bra{0} + i \ket{1}\!\!\bra{1}$ is the phase gate, and $\mathrm{CNOT}$ is the controlled-not gate.
We refer to unitary transformations in the Clifford group as Clifford circuits.
Clifford circuits with the addition of the $T$-gate are universal, where the $T$-gate is defined by $T \coloneqq \ket{0}\!\!\bra{0} + e^{i \pi/4} \ket{1}\!\!\bra{1}$.

A unitary transformation $U$ \textit{stabilizes} a state $\ket{\psi}$ when $U\ket{\psi} = \ket{\psi}$.
It is folklore that if an $n$-qubit state can be reached from the $\ket{0^n}$ state by applying a Clifford circuit, then the state is stabilized by a group of $2^n$ commuting members of the subset $\{\pm 1\} \times \{I, X, Y, Z\}^{\otimes n} \subset \left(\pauli \setminus -I^{\otimes n}\right)$, called its \textit{stabilizer group}. Such states are called \textit{stabilizer states}, and we denote the set of stabilizer states by $\stabset_n$. For $\ket{\psi} \in \stabset_n$, we denote its stabilizer group as $\stab(\ket{\psi})$. For more background on stabilizer states, see, e.g., \cite{nielsen2002quantum}.

We now define some complexity measures that characterize more general states in terms of stabilizer state decompositions.

\begin{definition}[stabilizer extent \cite{Bravyi2019simulationofquantum}]
Suppose $\ket{\psi}$ is a pure $n$-qubit state.
The \emph{stabilizer extent} of $\ket{\psi}$, denoted $(\ket{\psi})$, is the minimum of $\norm{c}_1^2$ over all decompositions $\ket{\psi} = \sum_i c_i \ket{\phi_i}$, where $\ket{\phi_i} \in \calS_n$ and $c$ is some vector in $\mathbb{C}^{|\stabset_n|}$.
\end{definition}
\begin{definition}[stabilizer fidelity \cite{Bravyi2019simulationofquantum}]
Suppose $\ket{\psi}$ is a pure $n$-qubit state.
The \emph{stabilizer fidelity} of $\ket{\psi}$, denoted $\fidelity_{\stabset}$, is 
$$
\fidelity_{\stabset}(\ket\psi) \coloneqq \max_{\ket{\phi} \in \stabset_n}\labs\braket{\phi|\psi}\rabs^2.
$$
\end{definition}

Below we give a useful relation between the complexity measures defined above.

\begin{claim}\label{claim:metric-relations}
Let $\ket{\psi}$ be an $n$-qubit pure state. Then,
\[
 \extent(\ket{\psi}) \geq \frac{1}{\fidelity_{\stabset}(\ket{\psi})} .\]

\end{claim}
\begin{proof}

Let $\ket{\psi} = \sum_{\ket{\phi} \in \stabset_n} c_{\phi} \ket{\phi}$ be such that $\left(\sum_\phi \abs{c_\phi}\right)^2 = \extent(\ket{\psi})$.
Suppose towards a contradiction that $\fidelity_{\stabset}(\ket{\psi}) <  \frac{1}{\extent(\ket{\psi})}$ and therefore $\abs{\braket{\phi|\psi}} <  \frac{1}{\extent(\ket{\psi})}$ for all $\ket{\phi} \in \stabset_n$.
Then,
\begin{align*}
1 = \abs{\braket{\psi|\psi}} = \labs \sum_{\ket{\phi_S} \in \stabset_n} c_{\phi}^\ast \braket{\phi|\psi} \rabs & \leq  \sum_{\ket{\phi_S} \in \stabset_n} \labs c_{\phi} \rabs \labs \braket{\phi|\psi} \rabs \\
& \leq \max_i \abs{\braket{\phi_i|\psi}} \sum_{\ket{\phi_S} \in \stabset_n} \labs c_{\phi} \rabs \\
& \le \fidelity_{\stabset}(\ket{\psi})\sqrt{\extent(\ket{\psi})}\\
& < \sqrt{\fidelity_{\stabset}(\ket{\psi})} \leq 1
\end{align*}
The last line follows from the fact that $\fidelity_{\stabset}(\ket{\psi}) \leq 1$ due to Cauchy-Schwarz and the definition of stabilizer fidelity. We then have $1 < 1$ as a clear contradiction.
\end{proof}

The claim above also follows as a special case of \cite[Theorem 4]{Bravyi2019simulationofquantum}, though its proof is more complicated.

To prove lower bounds on the number of $T$-gates necessary to prepare pseudorandom quantum states, we need to upper bound the stabilizer extent of a quantum state prepared by a Clifford+$T$ circuit comprised of $t$ $T$-gates.

\begin{claim}\label{claim:extent_linear_comb}
For $\ket{\psi} = \alpha \ket{v} + \beta \ket{w}$,
$$\extent(\ket{\psi}) \leq \left(\abs{\alpha} \sqrt{\extent(\ket{v})} + \abs{\beta} \sqrt{\extent(\ket{w})}\right)^2.$$
\end{claim}
\begin{proof}
Let $\ket{v} = \sum_i c_i \ket{\phi_i}$ and $\ket{w} = \sum_j d_j \ket{\varphi_j}$ be the minimal decompositions in terms of stabilizer extent (i.e., $\left(\sum_i \abs{c_i}\right)^2 = \extent(\ket{v})$).
Since $\ket{\psi} = \alpha \ket{v} + \beta \ket{w} = \alpha \sum_i c \ket{\phi_i} + \beta \sum_j d \ket{\varphi_j}$, we have a stabilizer decomposition of $\ket{\psi}$.
The stabilizer extent of this decomposition is at most 

\[
\left( \sum_i \abs{\alpha c_i + \beta d_i}\right)^2 \leq \left( \abs{\alpha} \sum_i \abs{c_i} + \abs{\beta}\sum_i \abs{d_i}\right)^2 \leq \left( \abs{\alpha} \sqrt{\extent(v)} + \abs{\beta}\sqrt{\extent(w)}\right)^2.\qedhere 
\]
\end{proof}

\begin{lemma}\label{lemma:clifford-t-extent}
Let $C$ be any Clifford+$T$ circuit comprised of $t$ $T$-gates and $\ket\psi = C \ket{0^n}$. Then, 

$$\extent(\ket{\psi}) \leq \left(1 + \frac{1}{\sqrt{2}}\right)^t.$$

\end{lemma}
\begin{proof}
We note that a Clifford+$T$ circuit can be broken into layers of Clifford circuits, followed by a single $T$-gate, followed by more layers of Clifford circuits, and so on.
Since Clifford circuits preserve stabilizer extent, we only need to show that the $T$-gate increases the stabilizer extent of any state by at most a constant multiplicative factor.
Since the SWAP gate is a Clifford operation, we assume without loss of generality that each $T$-gate is applied to the first qubit.

We proceed by induction on the layers of the circuit. In the first layer, when no $T$-gates have been applied, the bound is trivially true because the stabilizer extent of any stabilizer state is $1$.
Now, assume that, after applying some portion of the circuit $C^\prime$ to $\ket{0^n}$ with $t-1$ $T$-gates, we get the state $\ket{\varphi}$.
Observe that the $T$-gate can be expressed as $\cos(\pi/8)e^{i\pi/8} I + \sin(\pi/8)e^{i 13 \pi/8} Z$.
Thus, $(T \otimes I^{\otimes n-1}) \ket{\varphi} = \cos(\pi/8)e^{i\pi/8} \ket{\varphi} + \sin(\pi/8)e^{i 13 \pi/8} \left(Z\otimes I^{\otimes n-1}\right)\ket{\varphi}$.
Since $Z\otimes I^{\otimes n-1}$ is a Clifford operation, $\left(Z\otimes I^{\otimes n-1}\right) \ket{\varphi}$ has the same extent as $\ket{\varphi}$.
Therefore, applying \cref{claim:extent_linear_comb}, 
\[
\extent(\ket\psi) \leq \left(\cos(\pi/8) + \sin(\pi/8)\right)^2\extent(\ket{\varphi}) \leq \left(1 + \frac{1}{\sqrt{2}}\right)^{t}. \qedhere
\]
\end{proof}

\subsection{Boolean Fourier Analysis}\label{ssec:bool_fourier_analysis}

We review the basics of Fourier analysis over the Boolean hypercube.

\begin{definition}
Let $S \subseteq [n]$ be an index of bits. Then the \emph{parity function} on $S$ is defined to be
$$\chi_S(x) \coloneqq \prod_{i \in S}(-1)^{x_i}.$$
\end{definition}

Alternatively, we can define $\chi_S(x) = (-1)^{x \cdot s}$ where $s_i = 1$ if and only if $i \in S$.
This form will prove to be more natural for our purposes.

The parity functions are orthonormal under the inner product $\langle f, g\rangle = \frac{1}{2^n} \sum_{x \in \F_2^n} f(x)g(x)$.
Since there are $2^n$ distinct parity functions, this gives a complete basis.
Given a function $f: \F_2^n \to \mathbb{R}$, we can then write
$$
f(x) = \sum_{S \subseteq[n]}\hat{f}(S) \chi_S(x).
$$
The $\hat{f}(S)$ are real numbers known as the \textit{Fourier coefficients} (collectively known as the \textit{Fourier spectrum}), and are equivalently given by the formula
$$
\hat{f}(S) = \langle f(x), \chi_S(x) \rangle.
$$
As a basis change, we can then rethink inner products to be over the Fourier coefficients as well.
\begin{fact}[Plancherel's theorem]\label{fact:plancherel}
$$
\langle f, g \rangle = \frac{1}{2^n}\sum_{S \subseteq [n]} \hat{f}(S)\hat{g}(S).
$$
\end{fact}

Finally, the convolution is an operation that appears frequently in Fourier analysis over the reals. We can similarly define it over Boolean inputs.

\begin{definition}
For functions $f, g: \F_2^n \to \mathbb{R}$, we define the \emph{convolution} $f \ast g$ as
$$
(f \ast g)(x) \coloneqq \frac{1}{2^n}\sum_{t \in \F_2^n} f(t)g(x + t).
$$
\end{definition}

Much like Fourier transforms over the reals, convolution maps to multiplication in the Fourier domain.

\begin{fact}[Convolution theorem]\label{fact:convolution_theorem}
$\hat{f \ast g}(S) = \hat{f}(S)\hat{g}(S)$
\end{fact}

For proofs of all of these facts, as well as for a comprehensive reference on analysis of Boolean functions, we recommend \cite{o2014analysis}.

\subsection{Weyl Operators and Bell Difference Sampling}\label{ssec:weyl_and_bell}
For $x = (p,q) \in \F_2^{2n}$, define the \textit{Weyl operator} as 
\[W_x \coloneqq i^{p \cdot q} (X^{p_1}Z^{q_1}) \otimes \ldots \otimes (X^{p_n}Z^{q_n}) = i^{p \cdot q}X^pZ^q. \]
Each Weyl operator is a Pauli operator, and every Pauli operator is a Weyl operator (up to a phase). 
Note also that $W_x W_y = W_{x + y}$, up to a phase.
We use Weyl operators (rather than Pauli operators) when it is convenient to identify members of the Pauli group with length-$2n$ bit strings. 

A critical subroutine in our work is \textit{Bell difference sampling}, which was introduced in \cite{montanaro2017learning, gross2021schur}. 
Let $\ket{\Phi^+} \coloneqq 2^{-n/2} \sum_{x \in \F_2^n} \ket{x, x}$. Then, the set of quantum states $\{\ket{W_x} \coloneqq (W_x \otimes I) \ket{\Phi^+} \mid x\in\mathbb{F}_2^{2n}\}$ forms an orthonormal basis of $\mathbb{C}^{2^n} \otimes \mathbb{C}^{2^n}$, which we call the \textit{Bell basis}. 
Bell sampling a state $\ket{\psi}$ refers to measuring $\ket{\psi}^{\otimes 2}$ in the Bell basis, and the measurement outcome is a length-$2n$ bit string $x$ that corresponds to a Weyl operator $W_x$.
Bell difference sampling a state $\ket{\psi}$ refers to Bell sampling twice to get measurement outcomes $x, y \in \F_2^{2n}$ and returning $z = x + y$, which corresponds to a Weyl operator $W_z$ and uses four copies of $\ket{\psi}$.
Montanaro showed Bell difference sampling can be performed in $O(n)$ time \cite{montanaro2017learning}.

Bell difference sampling returns a random Weyl operator, but according to what distribution?
Gross, Nezami, and Walter\ \cite{gross2021schur} showed that the underlying distribution depends on the so-called characteristic distribution of $\ket{\psi}$.
\begin{definition}[characteristic distribution]
The \emph{characteristic distribution} of $\ket{\psi}$ is defined as
\[
p_{{\psi}}(x) \coloneqq 2^{-n} \abs{\braket{\psi|W_x|\psi}}^2.
\]
\end{definition}

\begin{lemma}[{\cite[Theorem 3.2]{gross2021schur}}]\label{lem:bell_diff_sampling}
Let $\ket{\psi}$ be an arbitrary $n$-qubit pure state. 
Bell difference sampling corresponds to drawing a sample from the following distribution:
\[
q_{\psi}(x) \coloneqq 4^n (p_\psi \ast p_\psi)(x) = \sum_{y \in \F_2^{2n}} p_{\psi}(y) p_{{\psi}}(x + y). 
\]
Additionally, if $\ket{\psi} \in \calS_n$ is a stabilizer state, then 
\[
q_{\psi}(x) = p_{{\psi}}(x) = 2^{-n} \abs{\braket{\psi|W_x|\psi}}^2.
\]
\end{lemma}

We refer to $q_\psi$ as the \textit{Weyl distribution}. 
Using this terminology, the characteristic distribution and Weyl distribution are equal only when $\ket{\psi}$ is a stabilizer state (i.e., when $4^n (p_\psi \ast p_\psi) = p_\psi$). 

\section{Certificate of Low Stabilizer Complexity}\label{sec:analysis}
To efficiently distinguish a state with low stabilizer complexity (meaning, a state with low stabilizer extent or non-negligible stabilizer fidelity) from a Haar-random one, we require a property or statistic of the state that distinguishes it from Haar-random.
As such, we present the following technical lemma, which forms the backbone of our algorithm.

\begin{lemma}
\label{lem:expected-measurement}
Let $\ket\psi$ be an $n$-qubit pure state.
If the stabilizer fidelity of $\ket{\psi}$ is at least $\frac{1}{k}$, then
$$
\Ex_{x \sim q_{\psi}}\lbrak \abs{\braket{\psi|W_x|\psi}}^2 \rbrak \geq \dfrac{1}{k^6}.
$$
In contrast, if $\ket{\psi}$ is Haar-random and $n \geq 33$, then, with probability at least $1-\exp\left(-2^{n/2-15}\right)$ over the Haar measure,
$$
\Ex_{x \sim q_{\psi}}\lbrak \abs{\braket{\psi|W_x|\psi}}^2 \rbrak \leq 2^{-n/2}.
$$
\end{lemma}

Our algorithm then amounts to estimating the quantity $\Ex_{x \sim q_{\psi}}\lbrak \abs{\braket{\psi|W_x|\psi}}^2 \rbrak$ via a procedure involving Bell difference sampling.

To prove \cref{lem:expected-measurement}, as a first step, we relate $\Ex_{x \sim q_{\psi}}\lbrak \abs{\braket{\psi|W_x|\psi}}^2 \rbrak$ to the Fourier coefficients of $p_\psi$.
Note that this analysis closely resembles the BLR linearity test \cite{BLRtest} (see also \cite[Section 1.6]{o2014analysis}). 

\begin{fact}\label{fact:blr-similarity}
Let $\ket{\psi}$ be an $n$-qubit pure state. Then, 
\[\Ex_{x \sim q_{\psi}}\lbrak \abs{\braket{\psi|W_x|\psi}}^2 \rbrak = 32^n \!\!\!\sum_{x \in \F_2^{2n}} \wh{p_\psi}(x)^3. \]
\end{fact}
\begin{proof}
\begin{align*}
\Ex_{x \sim q_{\psi}}\lbrak \abs{\braket{\psi|W_x|\psi}}^2 \rbrak 
&= 2^n \Ex_{x \sim q_{\psi}}\lbrak p_\psi(x) \rbrak \\
&= 2^n \sum_{x \in \F_2^{2n}} p_\psi(x) q_\psi(x) \\
&= 8^n \sum_{x \in \F_2^{2n}} p_\psi(x) (p_\psi \ast p_\psi)(x) \\
&= 32^n \Ex_{x \sim \F_2^{2n}}[p_\psi(x) (p_\psi \ast p_\psi)(x)] \\
&= 32^n \sum_{x \in \F_2^{2n}} \wh{p_\psi}(x) \wh{p_\psi\ast p_\psi}(x)) && (\mathrm{\cref{fact:plancherel}})\\
&= 32^n \sum_{x \in \F_2^{2n}} \wh{p_\psi}(x)^3. && (\mathrm{\cref{fact:convolution_theorem}}) \qedhere
\end{align*}
\end{proof}

For the remainder of this section, we use the following convention: when $x = (v, w) \in \F_2^{2n}$, $v$ and $w$ denote the first and last $n$ bits of $x$, respectively, and, we will sometimes write $p_\psi(v,w)$ and $q_\psi(v,w)$, rather than $p_\psi(x)$ and $q_\psi(x)$.

\subsection{The Fourier Spectrum of the Characteristic Distribution}
By \cref{fact:blr-similarity}, it is clear that understanding the Fourier spectrum of $p_\psi$ is one avenue to proving \cref{lem:expected-measurement}.

\begin{proposition}\label{prop:fourier-coefficients}
The Fourier coefficients of $p_\psi(v,w)$ are
$\wh{p_\psi}(v,w) = \frac{1}{2^n} p_\psi(w,v)$. 
\end{proposition}
\begin{proof}
Define $f: \F_2^{2n} \xrightarrow[]{} [-1,1]$ as $f(v,w) \coloneqq \braket{\psi|i^{v\cdot w}X^v Z^w |\psi}$, where $v,w \in \F_2^n$.  
We begin by computing the Fourier expansion of $f$.

\begin{align}\label{eq:expansion-of-f}
   f(v, w) 
   &= \bra{\psi}i^{v\cdot w} X^v Z^{w} \ket{\psi} \nn
   &=  \left( \sum_{x \in \F_2^n} c^*_x \bra{x} \right)i^{v\cdot w} X^v Z^{w}\left( \sum_{x \in \F_2^n} c_x \ket{x} \right) \nn
   &=  i^{v\cdot w}\left( \sum_{x \in \F_2^n} c^*_x \bra{x + v} \right) \left( \sum_{x \in \F_2^n} (-1)^{x \cdot w} c_x \ket{x} \right) \nn
   &=  i^{v\cdot w}\sum_{x \in \F_2^n} c_{x+v}^* c_x (-1)^{w \cdot x}.
\end{align}
In the second line we are simply writing $\ket{\psi}$ in the computational basis. 

Observe now that $p_\psi(v,w) = \frac{1}{2^n} \abs{f(v,w)}^2$, which we can also treat as a function on Boolean variables. Hence, 
\begin{align*}
   p_\psi(v,w) 
   &= \frac{1}{2^n}\left(i^{v\cdot w}\sum_{x \in \F_2^n} c_{x+v}^* c_x (-1)^{w \cdot x} \right)\left((-i)^{v\cdot w}\sum_{x \in \F_2^n} c_{x+v} c^*_x (-1)^{w \cdot x} \right) \nn
   &= \frac{1}{2^n}\sum_{x,y \in \F_2^n} c^*_{v+y}c_y c_{v+x+y}c^*_{x+y} (-1)^{w \cdot x},
\end{align*}
where the first equality follows by substituting in \cref{eq:expansion-of-f}.

We can now compute the Fourier spectrum of $p_\psi$ by taking the inner product between $p_\psi$ and an arbitrary Fourier character (this is the simplest approach to computing Fourier coefficients).
\begin{align*}
    \wh{p_\psi}(v,w) 
    &= \frac{1}{4^n}\sum_{s, t \in \F_2^n} p_\psi(s, t) (-1)^{s\cdot v + t\cdot w}\\
    &= \frac{1}{8^n}\sum_{s,t,x,y \in \F_2^n} c^*_{s+y}c_y c_{s+x+y}c^*_{x+y} (-1)^{t \cdot x + v \cdot s + w \cdot t}\\
    &= \frac{1}{8^n}\sum_{s,x,y \in \F_2^n} c^*_{s+y}c_y c_{s+x+y}c^*_{x+y} (-1)^{v \cdot s} \sum_{t\in\F_2^n}(-1)^{t \cdot (x+w)}\\
    &= \frac{1}{4^n}\sum_{s,y \in \F_2^n} c^*_{s+y}c_y c_{s+w+y}c^*_{w+y} (-1)^{v \cdot s} \\
    &= \frac{1}{2^n}p_\psi(w,v).\qedhere
\end{align*}
\end{proof}

\subsection{Low-Stabilizer-Complexity States}
We prove the first part of \cref{lem:expected-measurement}; namely, that
\[\Ex_{x \sim q_{\psi}}\lbrak \abs{\braket{\psi|W_x|\psi}}^2 \rbrak \geq \dfrac{1}{k^6}\]
when $\ket\psi$ has low stabilizer complexity.

\begin{claim} \label{lem:sum_p-hat_cubed}
For an $n$-qubit pure state $\ket{\psi} = \sum_{x \in \F_2^n} c_x \ket{x}$,  
\[
32^n \sum_{x \in \F_2^{2n}} \wh{p_\psi}(x)^3 \geq \abs{c_0}^{12}.
\]

\end{claim}
\begin{proof}
\begin{align*}
32^n\!\! \sum_{v,w \in \F_2^{2n}} \wh{p_\psi}(v,w)^3 
&= 4^n  \sum_{v,w \in \F_2^{2n}}\!\! p_\psi(w,v)^3 && (\mathrm{\cref{prop:fourier-coefficients}}.) \\
&\geq 4^n \sum_{v \in \F_2^n}p_\psi(0,v)^3 && (\forall x, y, p_\psi(x,y) \geq 0.)\\
&= \frac{1}{2^{n}} \sum_{v \in \F_2^n} \abs{\braket{\psi | Z^v  | \psi}}^{6}\\ 
& \geq \frac{1}{2^{6n}} \left(\sum_{v \in \F_2^n} \braket{\psi |Z^v| \psi}\right)^6  && \left(\sum_{i=1}^m \abs{a_i}^6 \geq \frac{1}{m^5}\left(\sum_{i = 1}^m \abs{a_i} \right)^6.\right)\\
& \geq \abs{c_0}^{12}.  && \left(\sum_{v \in \F_2^n} Z^v = 2^n \ketbra{0^n}{0^n}.\right)\qedhere
\end{align*}
\end{proof}

\begin{proof}[Proof of first part of \cref{lem:expected-measurement}]
Let $\ket{\psi}$ be an $n$-qubit pure state, and suppose that the stabilizer fidelity of $\ket{\psi}$ is at least $\frac{1}{k}$.
Then there exists a Clifford circuit $C \in \calC_n$ such that $C\ket{\psi} = \sum_{x \in \F_2^n} c_x \ket{x}$ where $\abs{c_{0}}^2 \geq \frac{1}{k}$. 
Call $\ket{\phi} \coloneqq C\ket{\psi}$.
By \cref{lem:sum_p-hat_cubed}, 
\[
32^n \sum_{v,w \in \F_2^n} \wh{p_{\phi}}(v, w)^3\geq \abs{c_0}^{12} \geq \frac{1}{k^6}.
\]
A Clifford circuit $C$ is a permutation of the Pauli group under conjugation (i.e., $C^\dagger \pauli C  = \pauli$ for any $C \in \calC_n$). Hence, for all $C \in \mathcal{C}_n$ and $g:\pauli \rightarrow \mathbb{R}$, $$\sum_{x \in \F_2^{2n}} g(W_x) = \sum_{x \in \F_2^{2n}} g(C^\dagger W_x C).$$
Therefore, we conclude that 
\[
32^n \sum_{v,w \in \F_2^n} \wh{p_{\psi}}(v, w)^3 \geq \frac{1}{k^6}
\]
as well. Combining this bound with \cref{fact:blr-similarity} completes the proof.
\end{proof}

\subsection{Haar-Random States}
We complete the proof of \cref{lem:expected-measurement} by showing that 
$\Ex_{x \sim q_{\psi}}\lbrak \abs{\braket{\psi|W_x|\psi}}^2 \rbrak$ is small when $\ket{\psi}$ is a Haar-random state. 
We begin by showing that, for a Haar-random state, all of the Weyl measurements (except $W_x = I$) are exponentially close to $0$ with overwhelming probability. 

\begin{lemma}[L\'{e}vy's Lemma, see e.g. \cite{Gerken13measureconcentration}]
\label{lem:levy}
Let $\mathbb{S}^N$ denote the set of all $N$-dimensional pure quantum states, and let $f: \mathbb{S}^N \to \mathbb{R}$ be $L$-Lipschitz, meaning that $\abs{f(\ket{\psi}) - f(\ket{\varphi})} \le L \cdot \norm{\ket{\psi} - \ket{\varphi}}_2$. Then:
\[
\Pr_{\ket{\psi} \sim \mu_{\rm{Haar}}} \left[\abs{f(\ket{\psi}) - \Ex[f]} \ge \eps  \right] \le 2\exp\left(-\frac{N \eps^2}{9\pi^3 L^2} \right).
\]

\end{lemma}

\begin{lemma}
\label{lem:weyl_lipschitz}
For any $n$-qubit Weyl operator $W_x$, the function $f_x: \mathbb{S}^{2^n} \to \mathbb{R}$ defined by $f_x(\ket{\psi}) = \bra{\psi}W_x\ket{\psi}$ is $2$-Lipschitz.
\end{lemma}

\begin{proof}
Write $W_x = \Pi_+ - \Pi_-$ where $\Pi_+$ and $\Pi_-$ are the projectors onto the positive and negative eigenspaces of $W_x$, respectively. 
Then,

\begin{align*}
    \abs{f_x(\ket{\psi}) - f_x(\ket{\varphi})}
    &= \abs{\bra{\psi}W_x\ket{\psi} - \bra{\varphi}W_x\ket{\varphi}}\\
    &= \abs{\bra{\psi}\Pi_+\ket{\psi} - \bra{\varphi}\Pi_+\ket{\varphi} - \bra{\psi}\Pi_-\ket{\psi} + \bra{\varphi}\Pi_-\ket{\varphi}}\\
    &\le \abs{\bra{\psi}\Pi_+\ket{\psi} - \bra{\varphi}\Pi_+\ket{\varphi}} + \abs{\bra{\psi}\Pi_-\ket{\psi} + \bra{\varphi}\Pi_-\ket{\varphi}}\\
    &= \abs{\ \norm{\Pi_+\ket{\psi}}_2 - \norm{\Pi_+\ket{\varphi}}}_2 + \abs{ \norm{\Pi_-\ket{\psi}}_2 - \norm{\Pi_-\ket{\varphi}}_2}\\
    &\le \norm{\Pi_+(\ket{\psi} - \ket{\varphi})}_2 + \norm{\Pi_-(\ket{\psi} - \ket{\varphi})}_2\\
    &\le 2\norm{\ket{\psi} - \ket{\varphi}}_2,
\end{align*}
where the third and fifth lines apply the triangle inequality, and the fourth and sixth lines use the fact that $\Pi_+$ and $\Pi_-$ are projectors.
\end{proof}

\begin{corollary}
\label{cor:single_weyl_concentration}
Let $W_x$ be any $n$-qubit Weyl operator in which $x \neq 0$ (i.e. $W_x \neq I$). Then:
\[
\Pr_{\ket{\psi} \sim \mu_{\rm{Haar}}} \left[\abs{\braket{\psi | W_x | \psi}} \ge \eps \right] \le 2\exp\left(-\frac{2^n \eps^2}{36 \pi^3} \right).
\]
\end{corollary}

\begin{proof}
    Define $f_x(\ket{\psi}) = \bra{\psi}W_x\ket{\psi}$ as in \Cref{lem:weyl_lipschitz}. By \Cref{lem:weyl_lipschitz}, we know that $f_x$ is $2$-Lipschitz. Additionally, observe that $\Ex[f] = 0$ over the Haar measure because exactly half of the eigenvalues of $W_x$ are $1$ and the other half are $-1$. Then the corollary follows from \Cref{lem:levy}.
\end{proof}

\begin{corollary}
\label{cor:all_weyl_concentration}
\[
\Pr_{\ket{\psi} \sim \mu_{\rm{Haar}}} \left[\exists x \neq 0 : \abs{\braket{\psi | W_x | \psi}} \ge \eps \right] \le 2^{2n + 1}\exp\left(-\frac{2^n \eps^2}{36 \pi^3} \right).
\]
\end{corollary}

\begin{proof}
    This follows from \Cref{cor:single_weyl_concentration} and a union bound over all $2^{2n}$ possible Weyl operators.
\end{proof}

Note that if $\eps \ge \frac{1}{\mathrm{poly}(n)}$, then the probability bound in \Cref{cor:all_weyl_concentration} is doubly-exponentially small.

We have shown that, with high probability, all Weyl measurements (except $W_x = I$) are close to $0$. 
We use this to complete the proof of \cref{lem:expected-measurement}.

\begin{proof}[Proof of second part of \cref{lem:expected-measurement}]
Suppose $\ket{\psi}$ is a Haar-random state.
By \cref{cor:all_weyl_concentration}, for all $W_x \neq I$, $\abs{\braket{\psi|W_x|\psi}}^2  = 2^n p(x) \leq \eps^2$ with probability $1-2^{2n+1}\exp\left(-\frac{2^n\eps^2}{36 \pi^3}\right)$. Therefore by \cref{fact:blr-similarity} and \cref{prop:fourier-coefficients}, 
\begin{align*}
\Ex_{x \sim q_{\psi}}\lbrak \abs{\braket{\psi|W_x|\psi}}^2 \rbrak 
&= 32^n \sum_{x, y \in \F_2^n} \widehat{p}(x, y)^3 \\
&= 4^n\sum_{w, v \in \F_2^n}p(v,w)^3\\
&= 4^n \left(\frac{1}{8^n} + \sum_{\substack{ w, v \in \F_2^n\\w,v \neq 0}} p(v,w)^3\right)\\
& \leq \frac{1 + (4^n-1)\eps^6}{2^n},
\end{align*}
with probability at least $1-2^{2n+1}\exp\left(-\frac{2^n\eps^2}{36 \pi^3}\right)$.
By setting $\epsilon^2 = \frac{1}{2^{n/6}}\left(\frac{2^n-2^{n/2}}{4^n-1}\right)^{1/3}$, we get 
\[
\Ex_{x \sim q_\psi}\lbrak \abs{\braket{\psi|W_x|\psi}}^2 \rbrak  \leq \frac{1}{2^{n/2}}
\]
with probability at least $1-2^{2n+1}\exp\left(-\frac{2^{5n/6}}{36 \pi^3}\left(\frac{2^n-2^{n/2}}{4^n-1}\right)^{1/3}\right)$, which is at least $1-\exp\left(-2^{n/2-15}\right)$ for $n \geq 33$.
\end{proof}
\section{Algorithm and Sample Complexity Analysis}\label{sec:main-result}

We are now ready to state and analyze our algorithm that distinguishes between Haar-random states and states with low stabilizer complexity. 
Our algorithm uses the fact that we can efficiently sample from $q_\psi$ (via Bell difference sampling) and efficiently estimate $\abs{\braket{\psi|W_x|\psi}}^2$ for any given $x \in \F_2^{2n}$, using quantum measurements.
By combining these subroutines, we construct an unbiased estimator for $\Ex_{x \sim q}\left[\abs{\braket{\psi| W_x |\psi}}^2\right]$. Motivated by \cref{lem:expected-measurement}, if our estimator exceeds a certain threshold we determine that the input state has low stabilizer complexity; otherwise, we determine that the state is Haar-random.
For the remainder of this section, $\eta \coloneqq  \Ex_{x \sim q}\left[\abs{\braket{\psi| W_x |\psi}}^2\right]$. 

\vspace{\baselineskip}
\begin{algorithm}[H]
\SetKwInOut{Promise}{Promise}
\caption{Distinguishing Low-Stabilizer-Complexity States from Haar-Random}
\label{alg:distinguisher}
\DontPrintSemicolon
\KwInput{Black-box access to copies of $\ket\psi$}
\Promise{$\ket{\psi}$ is Haar-random or has stabilizer fidelity at least $\frac{1}{k}$}
\KwOutput{$1$ if $\ket{\psi}$ has stabilizer fidelity at least $\frac{1}{k}$ and $0$ if $\ket{\psi}$ is Haar-random}
Let $m = 60k^{12}\ln(1/\delta)$.

\RepTimes{$m$}{
    Perform Bell difference sampling to obtain $W_x \sim q_{\psi}$.
    
    Perform the measurement $W_x^{\otimes 2}$ on $\ket{\psi}^{\otimes 2}$. Let $X_i \in \{\pm 1\}$ denote the measurement outcome.
}

Set $\hat{\eta} = \frac{1}{m} \sum_i X_i$. Return $1$ if $\hat{\eta} \geq \frac{2}{3k^6}$, and $0$ otherwise.
\end{algorithm}
\vspace{\baselineskip}

\begin{theorem}
\label{thm:main-thm-alg}
Let $\ket{\psi}$ be an unknown $n$-qubit pure state for some $n \geq 33$, and let $k \leq \frac{4}{5} 2^{n/12}$. \cref{alg:distinguisher} distinguishes whether $\ket{\psi}$ is Haar-random or a state with stabilizer fidelity at least $\frac{1}{k}$, promised that one of these is the case. 
The algorithm uses $O\left(k^{12}\log(1/\delta)\right)$ copies of $\ket{\psi}$ and $O(n k^{12} \log(1/\delta))$ time, and distinguishes the two cases with success probability at least $1-\delta$.
\end{theorem}

\begin{proof}
Following the notation in \cref{alg:distinguisher},
$X_i$ is the outcome of the measurement on the $i$th iteration of the algorithm loop.
Observe that for any $X_i$, 
\[
\Es{x\sim q_\psi,\\\text{meas. by $W_x^{\otimes 2}$}}[X_i] = \Ex_{x \sim q_\psi} \bra{\psi^{\otimes 2}}W_x^{\otimes 2}\ket{\psi^{\otimes 2}} = \Ex_{x \sim q_\psi} \abs{\braket{\psi|W_x|\psi}}^2= \eta.
\]
Therefore, $\hat{\eta} = \frac{1}{m} \sum_i X_i$ is an unbiased estimator of $\eta$ (i.e., $\Ex[\hat{\eta}] = \eta$).

Suppose $\ket{\psi}$ has stabilizer fidelity at least $\frac{1}{k}$. Then, our algorithm fails when $\hat{\eta} < \frac{2}{3k^6}$. Hence,
\[
\Pr[\text{\cref{alg:distinguisher} fails}] = \Pr\left[\hat{\eta} < \frac{2}{3k^6} \right] = \Pr\left[ \hat{\eta} - \eta < \frac{2}{3k^6} - \eta \right]  \leq \Pr\left[\hat{\eta} - \eta \leq -\frac{1}{3k^6}\right],
\]
where the last inequality follows from \cref{lem:expected-measurement}.
By Hoeffding's inequality, 
\[
\Pr\left[\hat{\eta} - \eta \leq -\frac{1}{3k^6}\right] \leq \exp\left(-\frac{m }{18 k^{12}} \right).
\]
Therefore, $m \geq 18k^{12} \ln(15) = 49k^{12}$ samples suffice for the failure probability to be at most $\frac{1}{15}$.

Now suppose $\ket{\psi}$ is Haar-random. Then, our algorithm fails when $\hat{\eta} \geq \frac{2}{3k^6}$.
By \cref{lem:expected-measurement}, $\eta \leq 2^{-n/2}$ with probability at least $1-\exp\left(-2^{n/2-15}\right) >= 1-e^{-2\sqrt{2}}$ for $n \geq 33$.
Assuming that $\eta \leq 2^{-n/2}$, 
\[
\Pr[\text{\cref{alg:distinguisher} fails}] = \Pr\left[\hat{\eta} \geq \frac{2}{3k^6} \right] = \Pr\left[ \hat{\eta} - \eta \geq \frac{2}{3k^6} - \eta \right]  \leq \Pr\left[\hat{\eta} - \eta \geq \frac{1}{2k^6} - 2^{-n/2} \right].
\]
Once again, by Hoeffding's inequality, 
\begin{align*}
\Pr\left[\hat{\eta} - \eta \geq \frac{1}{2k^6} - 2^{-n/2} \right] 
&\leq \exp\left(- \frac{m}{2} \left(\frac{2}{3k^6} - 2^{-n/2}\right)^2 \right) \\
&\leq \exp\left(- \frac{m}{2} \left(\frac{2}{3k^6} - \frac{1}{3k^6}\right)^2 \right) \\
&\leq \exp\left(- \frac{m}{18 k^{12}}\right).
\end{align*}
Therefore, $m \geq -18 k^{12} \ln\left(\frac{1}{15} - e^{-2\sqrt{2}}\right) \geq 88k^{12}$ samples suffice for the failure probability to be at most $\frac{1}{15} - e^{-2\sqrt{2}}$.   
By the union bound, the failure probability is at most $\frac{1}{15}$, where the randomness is over both the Haar measure and the quantum measurements.

We have shown that in either case we output the correct answer with probability at least $\frac{14}{15}$.
Using the Chernoff-Hoeffding theorem, the success probability can be boosted from $\frac{14}{15}$ to at least $1 - \delta$ by doing $\frac{2}{3}\ln (1/\delta)$ repetitions of \cref{alg:distinguisher} and taking the majority answer.
Since each iteration of the algorithm loop uses $6$ copies of $\ket{\psi}$, \cref{alg:distinguisher} consumes $O\left(k^{12} \log(1/\delta)\right)$ copies in total.
Finally, Bell difference sampling and performing the measurement $W_x^{\otimes 2}$ takes $O(n)$ time, so the total runtime is $O\left(n k^{12} \log(1/\delta)\right)$.
\end{proof}

All of these results also apply to states with stabilizer extent at most $k$, since by \cref{claim:metric-relations}, such states have stabilizer fidelity at least $\frac{1}{k}$.
\begin{corollary}
\label{cor:main-thm-alg-rank}
Let $\ket{\psi}$ be an unknown $n$-qubit pure state for $n \geq 33$, and let $k \leq \frac{4}{5} 2^{n/12}$. \cref{alg:distinguisher} distinguishes whether $\ket{\psi}$ is Haar-random or a state with stabilizer extent at most $k$, promised that one of these is the case. 
The algorithm uses $O\left(k^{12}\log(1/\delta)\right)$ copies of $\ket{\psi}$ and distinguishes the two cases with success probability at least $1-\delta$.
\end{corollary}

The above result immediately implies that output states of Clifford+$T$ circuits with few $T$-gates cannot be computationally pseudorandom.

\begin{corollary}
\label{cor:main-thm-pseudo}
Any family of Clifford+$T$ circuits that produces an ensemble of $n$-qubit computationally pseudorandom quantum states must use at least $\omega(\log n)$ $T$-gates.
\end{corollary}
\begin{proof}
Consider any ensemble of states wherein each state in the ensemble is the output of some Clifford+$T$ circuit with at most $K \log n$ $T$-gates. By \cref{lemma:clifford-t-extent}, the stabilizer extent of any such state $\ket{\psi}$ is at most $n^{\alpha K}$ for $\alpha \leq 0.7716$. 
By \cref{cor:main-thm-alg-rank}, on input copies of $\ket{\psi}$, \cref{alg:distinguisher} 
takes $O(n^{12\alpha K+1}) \le \poly(n)$ time and outputs $1$ with probability at least $2/3$. 
On the other hand, if $\ket{\psi}$ is a Haar-random state then the same algorithm outputs $1$ with probability at most $\frac{1}{3}$. As such, the algorithm's distinguishing advantage between the ensemble and the Haar measure is non-negligible. This is to say that the ensemble cannot be pseudorandom under the definition of \cite{Ji10.1007/978-3-319-96878-0_5}.
\end{proof}

\section{Open Problems}
\label{sec:discussion}

An immediate direction for future work is to improve the sample complexity of our algorithm, or to prove sample complexity lower bounds.
One can also endeavour to improve other features of our algorithm: Is it possible to remove the need for entangled measurements?\footnote{
The optimal algorithms for learning and testing stabilizer states use entangled measurements. So, a first step would be to understand how many separable measurements are required to separate stabilizer states from Haar-random.} Or, is it possible to show that entangled measurements are in any sense necessary?
Are there quantum measurements that allow us to sample from $p_\psi$ directly (rather than $q_\psi$)?

Beyond that, can Bell difference sampling be used for learning and/or property testing stabilizer-extent-$k$ states?
For stabilizer states ($k=1$), a $6$-query property testing algorithm is given by \cite{gross2021schur} and a $\Theta(n)$-query learning algorithm is given by \cite{montanaro2017learning}. Both algorithms rely on Bell difference sampling and are asymptotically optimal. We ask if there are generalizations of these algorithms for states with higher stabilizer complexity, similar to the question that was raised in \cite{arunachalam2022phase}.

\begin{question}
    Is there a $\poly(k)$-query algorithm for property testing stabilizer-extent-$k$ states? Likewise, is there a $\poly(n,k)$-time algorithm for learning stabilizer-extent-$k$ states?
\end{question}

Our results on stabilizer extent are due to the fact that extent and fidelity are inversely related. 
Is it possible to relate \textit{stabilizer rank} (a closely-related complexity measure, denoted by $\chi$) and stabilizer fidelity?
For instance, proving that, for all states $\ket{\psi}$,
\[
\fidelity_{\stabset}(\ket{\psi})^{-1} \leq \chi(\ket{\psi})^c,\quad\text{for any constant $c$,}
\] 
would imply that our algorithm can distinguish low-stabilizer-rank states from Haar-random. 
However, proving such a relation for even $c < \frac{\alpha n}{\log n}$ for $\alpha \leq 0.2284$ would imply super-linear lower bounds on the stabilizer rank of Clifford magic states, a long-standing open problem.

One can also ask if the lower bound on the number of $T$-gates necessary for computationally pseudorandom states can be improved.
\begin{question}
How many $T$-gates are necessary for a family of Clifford+$T$ circuits to produce an ensemble of $n$-qubit computationally pseudorandom states?
\end{question}

We remark that any improvements to our $\log n$ lower bound would require techniques beyond the ones used in our paper. 
Indeed, in \cref{sec:magic-extent} we show that one can hope for at most a quadratic improvement in the relationship between $\eta$ and stabilizer fidelity. Such an improvement would only yield constant-factor improvements on the number of $T$-gates necessary to prepare computationally pseudorandom states. 

On the other hand, we are not aware of any attempts to optimize the $T$-gate count for plausible constructions of $n$-qubit pseudorandom states. The best upper bound we know of is the essentially trivial bound of $O(n)$, based on constructions of with $O(n)$ general gates. This is because pseudorandom states can be constructed from pseudorandom functions (PRFs) with constant overhead \cite{brakerski10.1007/978-3-030-36030-6_10}, and PRFs are believed to be constructible in linear time \cite{ishai_10.1145/1374376.1374438, fan_10.1145/3519935.3520010}.\footnote{Technically, we are not sure whether the PRFs constructed in \cite{ishai_10.1145/1374376.1374438, fan_10.1145/3519935.3520010} are secure against quantum adversaries, which is necessary for instantiating \cite{brakerski10.1007/978-3-030-36030-6_10}'s construction, but we consider it reasonable to conjecture that linear-time quantum-secure PRFs exist.}

\section*{Acknowledgments}
We thank Scott Aaronson for helpful comments. SG, VI, DL are supported by Scott Aaronson's Vannevar Bush Fellowship from the US Department of Defense, the Berkeley NSF-QLCI CIQC Center, a Simons Investigator Award, and the Simons ``It from Qubit'' collaboration. WK is supported by an NDSEG Fellowship.
\bibliographystyle{alphaurl}
\bibliography{refs}

\newpage
\appendix

\section{Algorithm Improvements via State Preparation Unitary}\label{sec:qae}

When given access to a state preparation unitary for $\ket{\psi}$ (and its inverse), denoted by $U$ and $U^\dagger$, 
we can improve the sample and time complexities of our algorithm to $O\left( k^3 \log(1/\delta) \right)$ and $O\left( nk^3 \log(1/\delta) \right)$, respectively,  at the cost of $O\left( k^3 \log(1/\delta) \right)$ queries to $U$ and $U^\dagger$. 

Access to $U$ and $U^\dagger$ allows us to run quantum amplitude estimation (QAE) as a subroutine in our algorithm. 
Recall the well-known result of Brassard, H{\o}yer, Mosca, and Tapp:

\begin{theorem}[Quantum Amplitude Estimation (Theorem 12 in \cite{Brassard_2002})]\label{thm:qae}
Let $\Pi$ be a projector and $\ket{\psi}$ be an $n$-qubit pure state such that $\braket{\psi|\Pi|\psi} = \eta$.
Given access to the unitary transformations $R_{\Pi} = 2\Pi - I$ and $R_{\psi} = 2 \ket\psi\!\!\bra\psi - I$, there exists a quantum algorithm that outputs $\hat{\eta}$ such that 
\[\abs{\hat{\eta}- \eta} \leq \frac{2 \pi \sqrt{\eta(1-\eta)}}{m} + \frac{\pi^2}{m^2} \]
with probability at least $\frac{8}{\pi^2}$. The algorithm makes $m$ calls to $R_\Pi$ and $R_\psi$.
\end{theorem}

\begin{corollary}\label{cor:qae}
Let $\Pi$, $\ket{\psi}$, $R_\Pi$, and $R_\psi$ be the same as in \cref{thm:qae}.
There exists a quantum algorithm that outputs $\hat{\eta}$ such that 
\[\abs{\hat{\eta}- \eta} \leq \eps \]
with probability at least $\frac{8}{\pi^2}$. The algorithm makes no more than
\[\pi \frac{\sqrt{\eta (1-\eta) + \eps}}{\eps}\]
calls to $R_\Pi$ and $R_\psi$.
\end{corollary}
\begin{proof}
By \cref{thm:qae}, this will require 
$m$ queries, where $m$ is a solution to the following quadratic equation:
\[
\frac{2 \pi \sqrt{\eta(1-\eta)}}{m} + \frac{\pi^2}{m^2} \leq \eps \Rightarrow m \geq \pi \frac{\sqrt{\eta (1-\eta) + \eps}}{\eps} \geq\pi \frac{\sqrt{\eta(1-\eta)} +\sqrt{\eta(1-\eta) + \eps}}{2\eps}.
\]
\end{proof}

With that, we are ready to explain the modifications to \cref{alg:distinguisher} that achieves a quartic speedup in the dependency on $k$. 

\begin{theorem}
\label{thm:main-thm-alg-qae}
Let $\ket{\psi}$ be an unknown $n$-qubit pure state prepared by a unitary $U$ for $n \geq 33$, and let $k \leq \frac{4}{5} 2^{n/12}$. There exists a quantum algorithm that distinguishes whether $\ket{\psi}$ is Haar-random or a state with stabilizer fidelity at least $\frac{1}{k}$, promised that one of these is The case. 
The algorithm uses $O\left(k^{3}\log(1/\delta)\right)$ applications of either $U$ or $U^\dagger$ and time $O\left(n k^3 \log(1/\delta)\right)$, and distinguishes the two cases with success probability at least $1-\delta$.
\end{theorem}
\begin{proof}
We first define the \textit{Bell difference sampling projector} on $x$ as 
\[\Pi_x \coloneqq \sum_{y \in \F_2^{2n}} \ketbra{W_{y}}{W_y}\otimes \ketbra{W_{x+y}}{W_{x+y}}.\] 
Note that this is simply a compact way of writing the Bell difference sampling procedure: the probability of sampling $x$ is $q_\psi(x) = \lVert \Pi_x \ket{\psi^{\otimes 4}}
\rVert$.\footnote{Indeed, this is the way Gross, Nezami, and Walter \cite{gross2021schur} introduce Bell difference sampling.}
We can also perform the projective measurement $P_{\psi, x} \coloneqq W_x\ket{\psi}\!\!\bra{\psi} W_x = W_x U \ketbra{0}{0} U^\dagger W_x$, where this measurement is performed by applying $W_x$, $U^\dagger$, and then measuring in the computational basis.  
We can entangle $\Pi_x$ and $P_{\psi, x}$  to form the following projector: 
\[
M = \sum_{x \in \F_2^{2n}} \Pi_x \otimes P_{\psi, x}.
\]
Building $M$ involves controlled applications of $W_x$ according to the Bell difference sampling outcome. 
Observe that
\[
\braket{\psi^{\otimes 5}|M|\psi^{\otimes 5}} 
= \sum_{x \in \F_2^{2n}}\braket{\psi^{\otimes 4}|\Pi_x|\psi^{\otimes 4}} \cdot \braket{\psi|P_{\psi, x}|\psi}
= \Ex_{x \sim q_\psi} \lbrak \abs{\braket{\psi|W_x|\psi}}^2 \rbrak.  
\]

Hence, we can run QAE with the input projector $M$ and the input state $\ket{\psi^{\otimes5}}$, and the output will be an estimate of $\eta$ whose accuracy depends on $m$, the number of total calls to $R_\Pi$ and $R_\psi$. 

Proving the sample complexity bound will mimic \cref{thm:main-thm-alg}.
Suppose $\ket{\psi}$ is a state with stabilizer fidelity at least $\frac{1}{k}$. 
Define $\etamin \coloneqq \frac{1}{k^6}$, and note that for any state with stabilizer fidelity at least $\frac{1}{k}$, $\eta \geq \etamin$ due to \cref{lem:expected-measurement}.
For our algorithm to succeed, recall from the proof of \cref{thm:main-thm-alg} that we need 
\[
\abs{\hat{\eta} - \eta} \leq \abs{\frac{2}{3k^6} - \eta}.
\]
Therefore, we can run QAE with a fixed value of $m$ (to be specified later) for an estimate of $\eta$ whose accuracy is within $\pm \left(\eta - \frac{2}{3k^6}\right)$. 
By \cref{cor:qae}, 
\begin{align}\label{eq:qae-queries}
m \geq \pi \frac{\sqrt{\eta (1-\eta) + \eta - \frac{2}{3k^6}}}{\eta - \frac{2}{3k^6}}
\end{align}
queries suffice.
The chosen value of $m$ must work for all $\eta \in [\frac{1}{k^6}, 1]$. Note that \cref{eq:qae-queries} is monotonically decreasing for $\eta \in [\frac{2}{3k^6}, 1)$, and is therefore maximized by $\etamin$ for $\eta \in [\frac{1}{k^6} , 1]$. 
To succeed with probability at least $\frac{8}{\pi^2}$,
\[m \geq 4\pi k^3 \geq \pi \sqrt{12 k^6 - 9} = \pi \frac{\sqrt{\etamin (1-\etamin) + \etamin - \frac{2}{3k^6}}}{\etamin - \frac{2}{3k^6}}\] calls to $R_\Pi$ and $R_\psi$ suffices.

Now suppose $\ket\psi$ is a Haar-random state.
Again, by \cref{lem:expected-measurement}, we know that $\eta \leq 2^{-n/2}$ with probability $1-e^{-2\sqrt{2}}$ for $n \geq 33$.
Assuming $\eta \leq 2^{-n/2}$ and using \cref{cor:qae}, as long as we have
\[
m \geq \sqrt{6}\pi k^3 \geq \pi \frac{\sqrt{2^{-n/2} (1-2^{-n/2}) + \frac{2}{3k^6} - 2^{-n/2}}}{\frac{2}{3k^6} - 2^{-n/2}} \geq  \pi \frac{\sqrt{\eta (1-\eta) + \frac{2}{3k^6} - \eta}}{\frac{2}{3k^6} - \eta}
\]
queries to $R_{\Pi}$ and $R_\psi$, we obtain the correct answer with probability at least $\frac{8}{\pi^2}$.
In the inequalities above we use similar reasoning to the stabilizer fidelity $\frac{1}{k}$ case, combined with the fact that $2^{-n/2} \leq \frac{1}{3k^6}$.

Finally, since $R_\Pi$ and $R_\psi$ use a constant number of calls to $U$ and $U^\dagger$, 
the total number of calls is $O(k^3)$. 
Chernoff-Hoeffding can be used to bring the success probability from $3/4$ to $1-\delta$ using $6\ln(1/\delta)$ repetitions.
The runtime includes an extra factor of $O(n)$, due to the linear cost of both preparing $W_x$ and the Bell difference sampling projector, giving a $O\left(nk^3\log(1/\delta) \right)$ time complexity.
\end{proof}

\section{On the Tightness of Our Analysis}\label{sec:magic-extent}
We argue that the first part of \cref{lem:expected-measurement} is polynomially-close to optimal. We begin by computing the stabilizer extent and stabilizer fidelity of Clifford magic states.
The two technical ingredients involved in the computation are due to Bravyi et al.~\cite{Bravyi2019simulationofquantum}.

\begin{fact}[{\cite[Proposition 2]{Bravyi2019simulationofquantum}}]\label{prop:fidelty-extent-magic}
Let $\ket\psi$ be a Clifford magic state. Then, $\extent(\ket{\psi}) = \fidelity_{\stabset}(\ket{\psi})^{-1}$. 
\end{fact}

\begin{fact}[{\cite[Proposition 1]{Bravyi2019simulationofquantum}}]\label{prop:extent-multiplicative}
Let $\{\ket{\psi_1}, \ket{\psi_2}, \ldots, \ket{\psi_L}\}$ be any set of
states such that each state $\ket{\psi_j}$ describes a system of at most $3$ qubits. Then,
\[
\extent( \ket{\psi_1} \otimes \ket{\psi_2} \otimes \ldots \otimes \ket{\psi_L}) = \prod_i \xi(\ket{\psi_i}).
\]
\end{fact}

It is well known that the $m$-fold tensor product of $\ket{T} \coloneqq 2^{-1/2}( \ket{0} + e^{i \pi / 4}\ket{0})$ is a Clifford magic state.
Using the facts above, we can compute the stabilizer extent and stabilizer fidelity of $\ket{T^{\otimes m}}$.

\begin{fact}\label{fact:magic-extent}
\[\extent(\ket{T^{\otimes m}}) = \lpar \cos \pi/8 \rpar^{-2m}
\quad \text{and} \quad
F_{\stabset_m}(\ket{T^{\otimes m}}) = \lpar \cos \pi/8 \rpar^{2m}.
\]
\end{fact}

\begin{proof}
By \cref{prop:extent-multiplicative}, the stabilizer extent of $\ket{T^{\otimes{m}}}$ is simply the stabilizer extent of $\ket{T}$ raised to the power $m$. 
By \cref{prop:fidelty-extent-magic}, the stabilizer extent is the inverse of the stabilizer fidelity. Hence, the result follows simply by showing that the stabilizer fidelity of $\ket{T}$ is $\cos(\pi / 8)^2$, which can  be verified by explicit calculation over the $6$ different $1$-qubit stabilizer states. 
\end{proof}

Next, we compute $\eta$ for the state $\ket{T^{\otimes m}}$.

\begin{claim}\label{claim:eta-magic}
Let $\ket{\psi} = \ket{T^{\otimes m}}$ and define $\eta \coloneqq \Ex_{x \sim q_{\psi}}[2^n p_{\psi}(x)]$. Then, $\eta = (5/8)^m$.
\end{claim}
\begin{proof}
We begin by writing out $\ket{T}\!\!\bra{T}$ as a sum of Pauli matrices. By definition, 
$$
\ket{T}\!\!\bra{T} = \dfrac{1}{2}\lpar I + \dfrac{1}{\sqrt{2}}X + \dfrac{1}{\sqrt{2}}Y\rpar.
$$

We wish to compute $\sum_{x \in \F_2^{2m}}\widehat{p}_\psi(x)^3$.
We know that every such Pauli with nonzero $\widehat{p}_\psi(x)$ is a tensor product combination of $I$, $X$, and $Y$, so we enumerate over the number of indices where an $X$ or $Y$ appear.

\[
    \sum_{x \in \F_2^{2m}}\widehat{p}_{\psi}(x)^3 = \dfrac{1}{2^{6m}} \sum_{k=0}^m \binom{m}{k}\dfrac{1}{2^{3k}} \cdot 2^k = \dfrac{1}{64^m} \sum_{k=0}^m \binom{m}{k}\dfrac{1}{4^{k}} = \lpar \frac{5}{256}\rpar^m.
\]

Thus, by \cref{fact:blr-similarity},
\[
\eta = 32^m\sum_{x \in \F_2^{2m}}\widehat{p}_\psi(x)^3 = \lpar \frac{5}{8}\rpar^m. \qedhere
\]
\end{proof}

Combining \cref{claim:eta-magic} with \cref{lem:expected-measurement}, we have 
\[ \fidelity_{\stabset}(\ket{\psi}) \leq \eta^{1/c} = \left( \frac{5}{8}\right)^{m/c}\] 
for $c = 6$ (\cref{lem:expected-measurement}).
But, from \cref{fact:magic-extent}, we know that $\fidelity_{\stabset}(\ket{T^{\otimes m}}) = (\cos \pi/8)^{2m}$. Combining the two statements gives 
\[
(\cos \pi/8)^{2m} \leq (5/8)^{m/c}.
\]
$c \approx 2.97$ is the minimum $c$ that does not violate this inequality. 
Hence, one cannot hope for much more than a quadratic improvement in our bound.

\end{document}